\documentclass[12pt]{IEEEtran}
\IEEEoverridecommandlockouts
\usepackage{amsmath}
\usepackage{booktabs}
\usepackage[inline]{enumitem}
\usepackage[export]{adjustbox}
\usepackage{wrapfig}
\usepackage{subfigure}
\usepackage{comment}

\usepackage{amssymb,amsfonts}
\usepackage{amsthm}
\usepackage{commath}
\usepackage{amsmath}
\usepackage{algorithmic}
\usepackage{mathtools, nccmath}%

\usepackage{array}
\usepackage{textcomp}
\usepackage{xcolor}
\usepackage[final]{pdfpages}
\PassOptionsToPackage{bookmarks=false}{hyperref}
\def\BibTeX{{\rm B\kern-.05em{\sc i\kern-.025em b}\kern-.08em
    T\kern-.1667em\lower.7ex\hbox{E}\kern-.125emX}}

\theoremstyle{definition}
\newtheorem{theorem}{Theorem}

\theoremstyle{definition}

\theoremstyle{definition}

\theoremstyle{definition}

\theoremstyle{definition}

\theoremstyle{definition}

\begin{document}
\title{A Note on Low-Pass Filter Conditioning \\ for Current-Mode Control}

\onecolumn

\author{\IEEEauthorblockN{Xiaofan Cui and Al-Thaddeus Avestruz}
\thanks{X. Cui is with the Department of Electrical and Computer Engineering, University of California, Los Angeles, CA 90095, USA (e-mail: cuixf@seas.ucla.edu). A.-T. Avestruz is with the Department of Aerospace Engineering, University of Michigan, Ann Arbor, MI 48109, USA (e-mail: avestruz@umich.edu).}
}
\IEEEoverridecommandlockouts
\IEEEpubid{\makebox[\columnwidth]{\hfill} \hspace{\columnsep}\makebox[\columnwidth]{ }}
\maketitle
\IEEEpubidadjcol

\begin{abstract}
\!\!The low-pass filter is a classic control conditioning approach for high frequency current-mode control. However, to our knowledge, no prior work provides the large-signal stability criterion for the current-mode control with low-pass filters. This paper provides a mathematically rigorous large-signal stability criterion. The result can directly benefit the practical engineering implementation of the low-pass filter in high-frequency current-mode control.
\end{abstract}
\section{Introduction}
Low-pass filters, the traditional way to signal condition the interference, stabilize the current control loop by attenuating the amplitude of interference. 
However, as shown in Fig.\;\ref{fig:filter1}, the sensed current ramp may be distorted. Low-pass filters are traditionally not recommended to cut off before the switching frequency; however, this makes the filters unable to effectively suppress the interference whose spectrum is near or below the switching frequency. Because of our results on the robustness of peak current-mode control~\cite{cmpartone2022,cmparttwo2022,cui2023lcss}, we can allow the cut-off frequency of filters to be well below the switching frequency and still have good performance. Section~\ref{sec:cot} extends the results to constant on-time valley control. 

\section{System description}
We take constant off-time current-mode control as an example. The current control loop using constant off-time can be modeled as
\begin{align} 
    i_p[n] &= i_p[n-1] -m_2T_{\text{off}} +m_1\,t_{\text{on}}[n],\nonumber \\
    \label{eqn:filter_nonlin_sys_fb} i_c [n] &= h_0(T[n]) i_c[n-1] + \left(i_m(t) u(t) \ast h(t)\right)\bigg\rvert_{t=t_{\text{on}}[n]},
\end{align}
where ${\ast}$ is the convolution operator, $h(t)$ is the impulse response function of the low-pass filter, $h_0(t)$ is the zero input response of filter, and $u(t)$ is the unit step function; $T[n]=t_{\text{on}}[n]+T_{\text{off}}$ is the $n$th switching period.
Equation (\ref{eqn:filter_nonlin_sys_fb}) captures the low-pass filter response.
The variable $i_m(t)$ represents the current sensor output and filter input;
\begin{figure}[htp]
    \centering
    \includegraphics[width=8cm]{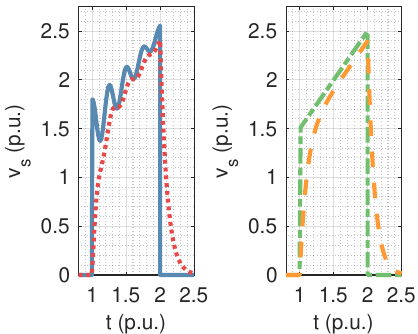}
    \caption{Comparison of current-sense voltage with/without the interference and with/without the filter. ({\color{blue}---}) is for the current-sense voltage with interference and without filter. ({\color{red}$\cdots$}) is for the current-sense voltage with interference and with filter. ({\color{green}-\,$\cdot$\,-}) is for the current-sense voltage without interference and without filter. ({\color{orange}-\,-\,-}) is for the current-sense voltage without interference with filter.}
    \label{fig:filter1}
\end{figure}
$i_m(t)$ can be expressed as the additive summation of the inductor current on the bottom switch and interference 
\begin{align}
    i_m(t) = i_p[n-1] - m_2T_{\text{off}} + m_1 t + w(t).
\end{align}
Here the current is sensed on the bottom switch of a boost converter, which conducts during the on interval. During the off interval, the sensed current, and hence the filter input, is zero, so the filter state only decays. This zero-input reset follows the circuit implementation.

\begin{figure}[htp]
\centering
\subfigure[Block diagram.]{\includegraphics[width=8 cm]{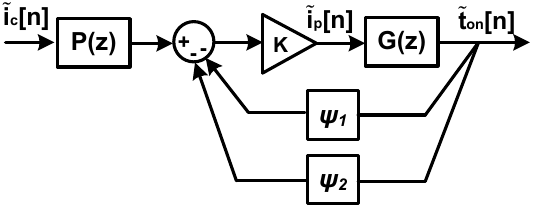}
       \label{fig:filteredccb_fb}
}
\subfigure[Root locus without (left)/with (right) filter. {\color{green}---} is for the negative feedback; {\color{blue} -\,-\,-} is for the positive feedback.]{
    \includegraphics[width=8 cm]{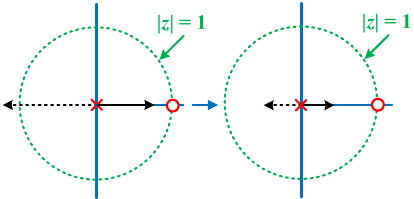}
    \label{fig:filter_cotcm}
}
\caption{Modeling of the current control loop with filter. The block is used in constant on(off)-time control.}
\end{figure}

We exemplify this idea with a first-order low-pass filter with the impulse response function $h(t) = (e^{(-t/\tau)}/\tau)\,u(t)$ and $h_0(t) = e^{-t/\tau}$, where $\tau$ represents the time constant. A continuous static mapping is a prerequisite for stability. Theorem\,\ref{theorem:taomin4continuity} provides a sufficient condition for the filter to guarantee a continuous static mapping. As long as the time constant satisfies Theorem\,\ref{theorem:taomin4continuity}, the static mapping is continuous. We denote the lower bound of the frequency of interference by $\omega_l$. The interference and its Fourier transform $W(\omega)$ are related by $w(t)=\frac{1}{2\pi}\int_{-\infty}^{+\infty}W(\omega)e^{j\omega t}\,d\omega$. We assume that $W(\omega)$ is absolutely integrable and that $W(\omega)=0$ for $|\omega|<\omega_l$, and we define the amplitude bound $A_{ub}\triangleq\frac{1}{2\pi}\int_{-\infty}^{+\infty}|W(\omega)|\,d\omega$, so that $|w(t)|\le A_{ub}$ for all $t$.
\begin{theorem} \label{theorem:taomin4continuity}
A current control loop using constant off\nobreakdash-time has minimum on time $T^{\text{min}}_{\text{on}}$ and off time $T_{\text{off}}$. 
The time constant of the first\nobreakdash-order low\nobreakdash-pass filter is $\tau$. 
The interference satisfies $W(\omega)=0$ for $|\omega|<\omega_l$ and has amplitude bound $A_{ub}$. The inductor current stays in $[0, I_{\text{max}}]$.
The condition for $\tau$ to guarantee a continuous static mapping is
\begin{align}
 \frac{\hat{A}_{ub}}{(1-d)\hat{\tau}}\left( 1 + \frac{d}{\sqrt{1 + (2\pi \hat{\omega}_l \hat{\tau})^2}}\right) + \frac{b\hat{I}_{\text{max}}}{(1-d)\hat{\tau}} < 1,
\end{align}
where
\begin{align}
    \hat{T}^{\text{min}} &= \hat{T}^{\text{min}}_{\text{on}} + \hat{T}_{\text{off}}, \quad
    b = e^{-\frac{\hat{T}^{\text{min}}}{\hat{\tau}}}, \quad d = e^{-\frac{\hat{T}^{\text{min}}_{\text{on}}}{\hat{\tau}}}, \quad \hat{\tau} = \frac{\tau}{T_{\text{on}}}, \quad \hat{T}_{\text{off}} = \frac{T_{\text{off}}}{T_{\text{on}}}, \nonumber \\ 
    \hat{A}_{ub} &= \frac{A_{ub}}{m_1T_{\text{on}}}, \quad
    \hat{I}_{\text{max}} = \frac{I_{\text{max}}}{m_1T_{\text{on}}},\quad 
    \hat{\omega}_l = \frac{\omega_l T_{\text{on}}}{2\pi}.
\end{align}
\end{theorem}
Here $T_{\text{on}}$ is the steady-state on time, which serves as the time base for normalization.
\begin{proof}
\subsection{Continuity Condition for a First-Order Filter}
Because the comparator cannot trip before the minimum on time, a sufficient condition for a continuous static mapping is that the filter output is strictly increasing for $t\ge T^{\text{min}}_{\text{on}}$. This condition can be expressed as
\begin{align} \label{eqn:cont_filter}
    & (i_p-m_2T_{\text{off}})\;\frac{\text{d}}{\text{d}t} \left[u(t) \ast h(t)\right]
     +  i_c h_0(T_{\text{off}})\; \frac{\text{d}h_0(t)}{\text{d}t} +
    \frac{\text{d}}{\text{d}t}\left\{\left[m_1t+w(t)\right]u(t)\ast h(t)\right\}
    >0,
\end{align}
for all $t\ge T^{\text{min}}_{\text{on}}$.

For all $t>0$, $u(t)\ast h(t)$ and $m_1t\,u(t)\ast h(t)$ are differentiable
\begin{align}
     \label{eqn:conv_eq1} \frac{\text{d}}{\text{d}t} \left[ m_1t\, u(t)\ast h(t) \right] &=  m_1u(t)\ast h(t), \\
     \label{eqn:conv_eq2} \frac{\text{d}}{\text{d}t}\left[ u(t)\ast h(t)\right]  &=  h(t).
\end{align}
Because the filter is first order, its natural response is determined by its single state, and the zero state response of the filter given input signal $w(t)u(t)$ is
\begin{align}  \label{eqn:fu_conv_h}
    w(t)u(t)\ast h(t) = \left[w(t)\ast h(t) - \frac{g(0)}{h_0(0)} h_0(t)\right]u(t).
\end{align}
where $\left[w(t)\ast h(t)\right] u(t)$ is the forced response and the rest is natural response.
We denote the forced response as $g(t) \triangleq w(t)\ast h(t)$.
For all $t>0$, $w(t)u(t)\ast h(t)$ is differentiable
\begin{align} \label{eqn:conv_eq_new} 
    \frac{\text{d}}{\text{d}t} \left[ w(t) u(t)\ast h(t) \right] = \frac{\text{d}g(t)}{\text{d}t} u(t)- \frac{g(0)}{h_0(0)}\;\frac{\text{d}h_0(t)}{\text{d}t} u(t).
\end{align}
By substituting (\ref{eqn:conv_eq1}), (\ref{eqn:conv_eq2}), (\ref{eqn:conv_eq_new}) into (\ref{eqn:cont_filter}), the continuity condition for the static mapping can be equivalently expressed as
\begin{align}
    & (i_p-m_2T_{\text{off}})h(t)
     +  i_c h_0(T_{\text{off}}) \; \frac{\text{d}h_0(t)}{\text{d}t}   + m_1\,u(t)\ast h(t) + \frac{\text{d}g(t)}{\text{d}t} u(t)- \frac{g(0)}{h_0(0)}\;\frac{\text{d}h_0(t)}{\text{d}t} u(t)
    >0.
\end{align}
\subsection{Continuity Theorem for First-Order Low-Pass Filter}
For a first-order low-pass filter with time constant $\tau$, the impulse response $h(t)$ and zero input response $h_0(t)$
\begin{align} \label{eqn:ht_def}
    h(t) = \frac{e^{-\frac{t}{\tau}}}{\tau}\,u(t),  \quad h_0(t) = e^{-\frac{t}{\tau}}\,u(t).
\end{align}
The frequency response function of the filter is
\begin{align} \label{eqn:hjw}
    H(j\omega) = \frac{1}{1+j\omega\tau}.
\end{align}
From (\ref{eqn:ht_def}) and (\ref{eqn:hjw}), we can bound $g$ from the top as
\begin{align}
    g'(t) = & \frac{1}{2\pi}\int_{-\infty}^{+\infty} j\omega W(\omega) H(\omega) e^{j\omega t}\,d\omega  =  \frac{1}{2\pi}\int_{-\infty}^{+\infty} \frac{j\omega W(\omega) }{1 + j\omega\tau} e^{j\omega t} \,d\omega.
\end{align}
\begin{align} \label{eqn:g_pr_t}
    |g'(t)| & \le \frac{1}{2\pi}\int_{-\infty}^{+\infty} \frac{|\omega|}{\sqrt{1 + \omega^2\tau^2}} |W(\omega)|  \,d\omega \le \frac{A_{ub}}{\tau}.
\end{align}
\begin{align}
    g(t) =  \frac{1}{2\pi}\int_{-\infty}^{+\infty} W(\omega) H(\omega) e^{j\omega t}\,d\omega 
    = \frac{1}{2\pi}\int_{-\infty}^{+\infty} \frac{W(\omega)}{1 + j \omega \tau} e^{j\omega t}\,d\omega.
\end{align}
\begin{align}
\label{eqn:s_lower_bound1_simp2}
    |g(t)| \le & \frac{1}{2\pi}\int_{-\infty}^{+\infty} |W(\omega)||H(\omega)|\,d\omega = \frac{1}{2\pi}\int_{-\infty}^{+\infty} \frac{|W(\omega)|}{\sqrt{1 + \omega^2\tau^2}} \,d\omega  \le \nonumber \\
    & \frac{1}{\sqrt{1 + (\omega_{l}\tau)^2}} \frac{1}{2\pi}\int_{-\infty}^{+\infty} |W(\omega)| \,d\omega \le \frac{A_{ub}}{\sqrt{1 + (\omega_{l}\tau)^2}}.
\end{align}
\begin{align} \label{eqn:g_zero_t}
 g(0) = & \frac{1}{2\pi}\int_{-\infty}^{+\infty} W(\omega) H(\omega) e^{j\omega t}\,d\omega \bigg \rvert_{t = 0}
 = \frac{1}{2\pi}\int_{-\infty}^{+\infty} \frac{W(\omega)}{1 + j\omega\tau} \,d\omega.
\end{align}
In the worst-case, $I_v = I_p - m_2T_{\text{off}} = 0$ . We substitute (\ref{eqn:g_pr_t}) and (\ref{eqn:s_lower_bound1_simp2}) into (\ref{eqn:cont_filter}). When $t \ge T^{\text{min}}_{\text{on}}$, (\ref{eqn:cont_filter}) can be bounded from below by

\begin{align} 
  & i_c h_0(T_{\text{off}}) \;\frac{\text{d}h_0(t)}{\text{d}t} + m_1\,u(t)\ast h(t)
    + \frac{\text{d}g(t)}{\text{d}t} u(t) - \frac{g(0)}{h_0(0)}\;\frac{\text{d}h_0(t)}{\text{d}t} u(t)  = \nonumber\\
    & -\frac{e^{-\frac{t}{\tau}}}{\tau}h_0(T_{\text{off}})I_c + m_1(1-e^{-\frac{t}{\tau}}) + g^{'}(t) + \frac{e^{-\frac{t}{\tau}}}{\tau} g(0) \ge \nonumber \\
     &
    -\frac{I_{\text{max}}}{\tau} \exp \left(-\frac{T^{\text{min}}_{\text{on}} + T_{\text{off}}}{\tau}\right)  + m_1(1-e^{-\frac{T^{\text{min}}_{\text{on}}}{\tau}}) - \frac{A_{ub}}{\tau} - \frac{e^{-\frac{T^{\text{min}}_{\text{on}}}{\tau}}}{\sqrt{1 + (\omega_l \tau)^2}} \frac{A_{ub}}{\tau} =  \nonumber \\
    & -\frac{b}{\tau} I_{\text{max}} + m_1(1-d) 
    - \frac{A_{ub}}{\tau} \left( 1 + \frac{d}{\sqrt{1 + (\omega_l\tau)^2}} \right)  \label{eqn:bd_slope},
\end{align}
where \mbox{$b = e^{-\frac{T^{\text{min}}}{\tau}}$}, $T^{\text{min}} = T^{\text{min}}_{\text{on}} + T_{\text{off}}$ and  $d = e^{-\frac{T^{\text{min}}_{\text{on}}}{\tau}}$.
As long as we impose a positive lower bound for (\ref{eqn:bd_slope}), the static mapping is continuous. The proof is complete.

\end{proof}

We next examine how the filter affects the stability of the current control loop.
At the operating point defined by the desired peak inductor current $I_c$, the actual peak inductor current $I_p$, the actual valley inductor current $I_v$, at on time $T_{\text{on}}$, we linearize system (\ref{eqn:filter_nonlin_sys_fb}) as
\begin{align} \label{eqn:filter_lin_sys2}
\tilde i_p[n] = \; & \tilde i_p[n-1]  +m_1 \, \tilde t_{\text{on}}[n], \nonumber \\
\tilde i_c [n] = \; & h_0(T_{\text{on}}) h_0(T_{\text{off}}) \,\tilde i_c [n-1] +  c_1 \, \tilde i_p [n] +  c_2 \, \tilde t_{\text{on}}[n],
\end{align}
where 
\begin{align} \label{eqn:prarameterszdomain}
    c_1 = &  u(t) {\ast} h(t)\bigg\rvert_{t = T_{\text{on}}}, \nonumber \\
    c_2 = & \frac{\text{d}h_0\,(t)}{\text{d}t}\bigg\rvert_{t = T_{\text{on}}} h_0(T_{\text{off}})I_c + h(T_{\text{on}})I_v 
    +\;\frac{\text{d}}{\text{d}t}\left[w(t)u(t){\ast}h(t)\right] \bigg\rvert_{t = T_{\text{on}}}.
\end{align}
System (\ref{eqn:filter_lin_sys2}) is represented by the block diagram in Fig.\;\ref{fig:filteredccb_fb}. The detailed derivations are provided in Appendix \ref{sec:proofoptsc}.

The gain term $K$, pole-zero pair, and feedback gains $\psi_1$ and $\psi_2$ are introduced as
\begin{align}
K = \frac{1}{1-e^{-\frac{T_{\text{on}}}{\tau}}}.
\end{align}
\begin{align}
    P(z) = 1 - e^{-\frac{T}{\tau}}z^{-1}.
\end{align}
\begin{align}
     \psi_1 = \frac{\text{d}g(t)}{\text{d}t}\bigg\rvert_{t = T_{\text{on}}} + \frac{e^{-\frac{T_{\text{on}}}{\tau}}}{\tau}g(0).
\end{align}
\begin{align}
     \psi_2 = - \frac{e^{-\frac{T}{\tau}}}{\tau}I_c + \frac{e^{-\frac{T_{\text{on}}}{\tau}}}{\tau}(I_p - m_2T_{\text{off}}).
\end{align}

Theorem \ref{theorem:taomin4stability} provides the condition for $\tau$ so that the current control loop is globally asymptotically stable.
If $\tau$ satisfies the condition, the globally asymptotic stability is guaranteed.

\begin{theorem} \label{theorem:taomin4stability}
A current control loop using constant off\nobreakdash-time has a minimum on time $T_{\text{on}}^{\text{min}}$ and fixed off time $T_{\text{off}}$. The time constant of the first\nobreakdash-order, low\nobreakdash-pass filter is $\tau$.
The interference satisfies $W(\omega)=0$ for $|\omega|<\omega_l$ and has amplitude bound $A_{ub}$. The inductor current stays in $[0, I_{\text{max}}]$.
The bound on $\tau$ to guarantee the globally asymptotic stability of the current control loop is
\begin{align} \label{thoerem:filterstab_v2}
 k_4 \frac{\hat{I}_{\text{max}}}{\hat{\tau}} + k_1 \frac{\hat{A}_{ub}}{\hat{\tau}}+ k_2 \frac{\hat{A}_{ub}}{\hat{\tau}\sqrt{1 + (2\pi\hat{\omega}_l\hat{\tau})^2} } < \frac{1}{2},
\end{align}
where
\begin{align}
    k_1 & = \frac{1}{1-d}, \quad k_2 = \frac{1+d^2}{(1-d)^2}, \quad k_4 = \frac{d}{(1-d)^2}, \quad d = e^{-\frac{\hat{T}^{\text{min}}_{\text{on}}}{\hat{\tau}}}, \nonumber \\
    \hat{\tau} & = \frac{\tau}{T_{\text{on}}}, \quad \hat{A}_{ub} = \frac{A_{ub}}{m_1T_{\text{on}}}, \quad \hat{I}_{\text{max}} = \frac{I_{\text{max}}}{m_1T_{\text{on}}}, \quad \hat{\omega}_l = \frac{\omega_l T_{\text{on}}}{2\pi}.
\end{align}
\end{theorem}
\begin{proof}
\begin{figure}
    \centering
    \includegraphics[width  = 10cm]{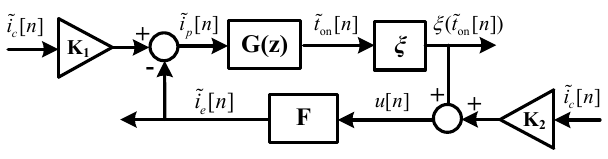}
    \caption{Block diagram of the current control loop with filter.}
    \label{fig:model_filter_nonlinear}
\end{figure}
The current control loop with filter can be modeled as the block diagram in Fig.\,\ref{fig:model_filter_nonlinear}, where
\begin{align}
    K_1  = \, & e^{-\frac{T_{\text{off}}}{\tau}}, \\
    K_2 = \, & (e^{-\frac{T_{\text{off}}}{\tau}}-1),  \\
    \xi(\tilde{t}_{\text{on}}) = \, & \left[I_c e^{-\frac{T_{\text{off}}}{\tau}} - (I_p - m_2T_{\text{off}}) + m_1\tau \right] \left[ \exp{\left(-\frac{T_{\text{on}}+\tilde{t}_{\text{on}}}{\tau}\right)} - \exp{\left(-\frac{T_{\text{on}}}{\tau}\right)}\right] \nonumber \\
    & + \left[w(t)\, u(t) \ast h(t)\right]\bigg\rvert^{t =
    \tilde{t}_{\text{on}} + 
    T_{\text{on}}}_{t = T_{\text{on}}}.
\end{align}
From the small gain theorem,
\begin{align} \label{eqn:ss_gf}
    \|G\|_2 \cdot \|\xi\|_2 \cdot \|F\|_2 < 1.
\end{align}
$\tilde{i}_e[n]$ is defined as
\begin{align}
    \tilde{i}_e[n] = e^{-\frac{T_{\text{off}}}{\tau}} \tilde{i}_c[n] - \tilde{i}_p[n].
\end{align}
The nonlinear subsystem $F$ follows
\begin{align}
    \tilde{i}_e[n] = \exp{\left(-\frac{\tilde{t}_{\text{on}}[n]+T_{\text{on}}}{\tau}\right)} \tilde{i}_e[n-1] + u[n].
\end{align}
We can find the $L_2$ gain of the subsystem $G$ and $F$ as
\begin{align}
    \|G\|_2 \le & \frac{2}{m_1},
    \\
    \|F\|_2 \le & \frac{1}{1-e^{-\frac{T^{\text{min}}_{\text{on}}}{\tau}}}.
\end{align}
We can find the $L_2$ gain of the nonlinear function $\xi$ by evaluating its derivatives 
\begin{align}
    \frac{\text{d}\xi}{\text{d}\tilde{t}_{\text{on}}} =  \psi_1(t_{\text{on}}) + \psi_2(t_{\text{on}}),
\end{align}
where
\begin{align}
    \label{eqn:psi1_ton} \psi_1(t_{\text{on}}) & = \frac{\text{d}g(t)}{\text{d}t}\Bigg\rvert_{t = t_{\text{on}}} + \frac{e^{-\frac{t_{\text{on}}}{\tau}}}{\tau}g(0),  \\
     \label{eqn:psi2_ton} \psi_2(t_{\text{on}}) & = 
    -\frac{1}{\tau}\left(I_c e^{-\frac{T_{\text{off}}}{\tau}} - (I_p - m_2T_{\text{off}}) + m_1\tau\right)e^{-\frac{t_{\text{on}}}{\tau}},\\
    g(t) & = \frac{1}{2\pi}\int_{-\infty}^{+\infty} \frac{W(\omega)}{1 + j \omega \tau} e^{j\omega t}\,d\omega.
\end{align}
$\psi_1(t_{\text{on}})$ can be bounded from above as 
\begin{align}
    |\psi_1(t_{\text{on}})| & <  \frac{A_{ub}}{\tau} + \frac{1}{\sqrt{ 1 + \omega_l^2\tau^2}} \frac{A_{ub}}{\tau} \triangleq \psi_1^{\text{max}}.
\end{align}
To bound $\psi_2(t_{\text{on}})$, we use the equilibrium relation (\ref{eqn:equic_app}) at the operating point $(I_c, I_p, T_{\text{on}})$, where it holds. Writing $I_v = I_p - m_2T_{\text{off}}$, it gives
\begin{align} \label{eqn:equic}
    I_c e^{-\frac{T_{\text{off}}}{\tau}} - I_v + m_1\tau = \frac{m_1T_{\text{on}} - \left(1-e^{-\frac{T_{\text{off}}}{\tau}}\right)I_c + g(T_{\text{on}}) - g(0)e^{-\frac{T_{\text{on}}}{\tau}}}{1-e^{-\frac{T_{\text{on}}}{\tau}}}.
\end{align}
Substituting (\ref{eqn:equic}) into (\ref{eqn:psi2_ton}) gives
\begin{align} \label{eqn:psi2_eq}
    \psi_2(t_{\text{on}}) = -\frac{e^{-\frac{t_{\text{on}}}{\tau}}}{\left(1-e^{-\frac{T_{\text{on}}}{\tau}}\right)\tau}\left[m_1T_{\text{on}} - \left(1-e^{-\frac{T_{\text{off}}}{\tau}}\right)I_c + g(T_{\text{on}}) - g(0)e^{-\frac{T_{\text{on}}}{\tau}}\right].
\end{align}
Here $t_{\text{on}}$ enters only through $e^{-t_{\text{on}}/\tau}$, so the equilibrium relation is not used away from the operating point. Because the inductor current stays in $[0, I_{\text{max}}]$, we have $0 \le I_c \le I_{\text{max}}$ and $m_1T_{\text{on}} = I_p - I_v \le I_{\text{max}}$. With $t_{\text{on}}, T_{\text{on}} \ge T^{\text{min}}_{\text{on}}$ and (\ref{eqn:s_lower_bound1_simp2}), $\left|g(T_{\text{on}}) - g(0)e^{-T_{\text{on}}/\tau}\right| \le (1+d)A_{ub}/\sqrt{1+\omega_l^2\tau^2}$. Therefore,
\begin{align}
    \psi_2(t_{\text{on}}) & \ge -\frac{d}{(1-d)\tau}\left(I_{\text{max}} + (1+d)\frac{A_{ub}}{\sqrt{1+\omega_l^2\tau^2}}\right) \triangleq \psi_2^{\text{min}}, \\
    \psi_2(t_{\text{on}}) & \le \frac{d-b}{(1-d)\tau}I_{\text{max}} + \frac{(1+d)d}{(1-d)\tau}\frac{A_{ub}}{\sqrt{1+\omega_l^2\tau^2}} \triangleq \psi_2^{\text{max}},
\end{align}
where $b = e^{-T^{\text{min}}/\tau}$, $T^{\text{min}} = T^{\text{min}}_{\text{on}} + T_{\text{off}}$, and $d = e^{-T^{\text{min}}_{\text{on}}/\tau}$. Because $d-b < d$, $\psi_2^{\text{max}} < -\psi_2^{\text{min}}$.
We denote the upper bound of the $L_2$ gain of $\xi$ by $B_{\xi}$
\begin{align}
    B_{\xi} = \text{max}\left\{
    \psi_{1}^{\text{max}}-\psi_{2}^{\text{min}},   \psi_{2}^{\text{max}}+\psi_{1}^{\text{max}}
    \right\}.
\end{align}
From (\ref{eqn:ss_gf}), the stability of the current control loop is guaranteed if the filter $\tau$ satisfies the condition
\begin{align}
    \frac{2}{m_1}\cdot\frac{1}{1-d}\cdot B_{\xi} < 1.
\end{align}
Because $\psi_2^{\text{max}} < -\psi_2^{\text{min}}$, we have $B_{\xi} = \psi_1^{\text{max}} - \psi_2^{\text{min}}$. Substituting $\psi_1^{\text{max}}$ and $\psi_2^{\text{min}}$ into the condition above and dividing by $m_1(1-d)/2$ gives (\ref{thoerem:filterstab_v2}).
\end{proof}

\section{Constant On-Time Valley Control} \label{sec:cot}
The results above are derived for constant off-time peak control. Constant on-time valley control is its dual: the on interval $T_{\text{on}}$ is fixed, the off interval $t_{\text{off}}[n]$ is variable, and the comparator ends the off interval when the filtered current falls to the valley command $i_c[n]$. The current is sensed on the bottom switch of a buck converter, which conducts during the off interval. During the fixed on interval the filter input is zero, so the filter state only decays. Measuring $t$ from the turn-off instant, the current control loop is
\begin{align}
    i_p[n] = i_v[n-1] + m_1T_{\text{on}}, \quad i_v[n] = i_p[n] - m_2\,t_{\text{off}}[n], \label{eqn:cot_model_i}\\
    i_c[n] = h_0(T[n])\, i_c[n-1] + \left(i_m(t)u(t)\ast h(t)\right)\big\rvert_{t=t_{\text{off}}[n]}, \label{eqn:cot_model_c}
\end{align}
where $i_m(t) = i_p[n] - m_2t + w(t)$ and $T[n] = T_{\text{on}} + t_{\text{off}}[n]$. Compared with (\ref{eqn:filter_nonlin_sys_fb}), the on and off intervals and the slopes $m_1$ and $m_2$ exchange roles. The sensed ramp now falls, and each off interval starts from the decayed valley command while the filter is driven by the peak current. Because of this, the sign arguments in the proofs of Theorems~\ref{theorem:taomin4continuity} and~\ref{theorem:taomin4stability} do not carry over term by term, and the conditions below differ from those obtained by simply exchanging $m_1$ and $m_2$. In this section the steady-state off time $T_{\text{off}}$ is the time base: $\hat{\tau} = \tau/T_{\text{off}}$, $\hat{A}_{ub} = A_{ub}/(m_2T_{\text{off}})$, $\hat{I}_{\text{max}} = I_{\text{max}}/(m_2T_{\text{off}})$, $\hat{\omega}_l = \omega_lT_{\text{off}}/(2\pi)$, $\hat{T}^{\text{min}}_{\text{off}} = T^{\text{min}}_{\text{off}}/T_{\text{off}}$, $\hat{T}_{\text{on}} = T_{\text{on}}/T_{\text{off}}$, $d = e^{-\hat{T}^{\text{min}}_{\text{off}}/\hat{\tau}}$, and $b = e^{-(\hat{T}^{\text{min}}_{\text{off}}+\hat{T}_{\text{on}})/\hat{\tau}}$.

\begin{theorem} \label{theorem:cot_continuity}
A current control loop using constant on-time valley control has a fixed on time $T_{\text{on}}$ and a minimum off time $T^{\text{min}}_{\text{off}}$, and uses a first-order low-pass filter with time constant $\tau$. The interference satisfies $W(\omega)=0$ for $|\omega|<\omega_l$ and has amplitude bound $A_{ub}$, and the inductor current stays in $[0, I_{\text{max}}]$. The static mapping is continuous if
\begin{align}
    \frac{d}{1-d}\,\frac{\hat{I}_{\text{max}}}{\hat{\tau}} + \frac{\hat{A}_{ub}}{(1-d)\hat{\tau}}\left(1 + \frac{d}{\sqrt{1+(2\pi\hat{\omega}_l\hat{\tau})^2}}\right) < 1. \label{thm:cot_cont}
\end{align}
\end{theorem}
\begin{proof}
The valley event is the first falling crossing of the filter output $y(t)$ for $t \ge T^{\text{min}}_{\text{off}}$, so a sufficient condition is that $y(t)$ is strictly decreasing for $t \ge T^{\text{min}}_{\text{off}}$. Following the steps of the proof of Theorem~\ref{theorem:taomin4continuity},
\begin{align}
    \frac{\text{d}y}{\text{d}t} = -\frac{e^{-\frac{t}{\tau}}}{\tau}h_0(T_{\text{on}})\,i_c[n-1] + \frac{e^{-\frac{t}{\tau}}}{\tau}i_p[n] - m_2\left(1-e^{-\frac{t}{\tau}}\right) + \frac{\text{d}g(t)}{\text{d}t} + \frac{e^{-\frac{t}{\tau}}}{\tau}g(0).
\end{align}
The first term is not positive, $i_p[n] \le I_{\text{max}}$, and (\ref{eqn:g_pr_t}) and (\ref{eqn:s_lower_bound1_simp2}) bound the last two terms. The resulting upper bound decreases with $t$, so its largest value is at $t = T^{\text{min}}_{\text{off}}$:
\begin{align}
    \frac{\text{d}y}{\text{d}t} \le \frac{1}{\tau}\left(d\,I_{\text{max}} - m_2\tau(1-d) + A_{ub} + d\,\frac{A_{ub}}{\sqrt{1+\omega_l^2\tau^2}}\right).
\end{align}
This bound is negative if and only if (\ref{thm:cot_cont}) holds. In Theorem~\ref{theorem:taomin4continuity} the corresponding coefficient is $b$ rather than $d$: there the valley term is non-negative and the filter starts from the decayed command, whereas here the filter is driven by the peak current from the start of the off interval.
\end{proof}

\begin{theorem} \label{theorem:cot_stability}
Under the assumptions of Theorem~\ref{theorem:cot_continuity}, the current control loop using constant on-time valley control is globally asymptotically stable if
\begin{align}
    k_5\,\frac{\hat{I}_{\text{max}}}{\hat{\tau}} + k_1\frac{\hat{A}_{ub}}{\hat{\tau}} + k_2\frac{\hat{A}_{ub}}{\hat{\tau}\sqrt{1+(2\pi\hat{\omega}_l\hat{\tau})^2}} < \frac{1}{2}, \qquad k_5 = \frac{2d-b}{(1-d)^2}, \label{thm:cot_stab}
\end{align}
where $k_1 = 1/(1-d)$ and $k_2 = (1+d^2)/(1-d)^2$ as in Theorem~\ref{theorem:taomin4stability}, with $d$ and $b$ defined above.
\end{theorem}
\begin{proof}
The structure of Fig.~\ref{fig:model_filter_nonlinear} applies with $\tilde{t}_{\text{off}}$ in place of $\tilde{t}_{\text{on}}$, so $\|G\|_2 \le 2/m_2$ and $\|F\|_2 \le 1/(1-d)$. Repeating the derivation of (\ref{eqn:psi2_ton}) for the falling ramp gives
\begin{align}
    \psi_2(t_{\text{off}}) = \frac{1}{\tau}\left(I_p + m_2\tau - e^{-\frac{T_{\text{on}}}{\tau}}I_c\right)e^{-\frac{t_{\text{off}}}{\tau}},
\end{align}
and $\psi_1$ is unchanged with $t_{\text{off}}$ in place of $t_{\text{on}}$. Using the equilibrium relation at the operating point $(I_c, I_p, T_{\text{off}})$,
\begin{align}
    \psi_2(t_{\text{off}}) = \frac{e^{-\frac{t_{\text{off}}}{\tau}}}{\left(1-e^{-\frac{T_{\text{off}}}{\tau}}\right)\tau}\left[\left(1-e^{-\frac{T_{\text{on}}}{\tau}}\right)I_c + m_2T_{\text{off}} - g(T_{\text{off}}) + g(0)e^{-\frac{T_{\text{off}}}{\tau}}\right].
\end{align}
Compared with (\ref{eqn:psi2_eq}), the command term and the ramp term now have the same sign, so $\psi_2$ is mostly positive: seen through the filter, the falling ramp adds a positive-feedback path. With $0 \le I_c \le I_{\text{max}}$ and $m_2T_{\text{off}} = I_p - I_v \le I_{\text{max}}$,
\begin{align}
    -\frac{(1+d)d}{(1-d)\tau}\frac{A_{ub}}{\sqrt{1+\omega_l^2\tau^2}} \le \psi_2(t_{\text{off}}) \le \frac{2d-b}{(1-d)\tau}I_{\text{max}} + \frac{(1+d)d}{(1-d)\tau}\frac{A_{ub}}{\sqrt{1+\omega_l^2\tau^2}}.
\end{align}
Here the upper bound has the larger magnitude, so $B_{\xi} = \psi_2^{\text{max}} + \psi_1^{\text{max}}$, and the small-gain condition $\frac{2}{m_2}\cdot\frac{1}{1-d}\cdot B_{\xi} < 1$ gives (\ref{thm:cot_stab}).
\end{proof}

\section{Conclusion}
By applying the switching-synchronized sampled-state space (5S) model~\cite{Cui2018a,xiaofanacc2019,cui2023tpel} and Small Gain Theorem \cite{Okuyama2014}, this paper develops two theoretical results for high-frequency current-mode control using low-pass filters: (1) the continuity condition of the static mapping; (2) a large-signal stability criterion of the dynamical mapping.
The results allow the cut-off frequency of filters to be well below the switching frequency and still have good performance. Section~\ref{sec:cot} extends both results to constant on-time valley control.
\newpage
\begin{appendices}
\section{Linearized Model of the Current Control Loop with Low-Pass Filter} \label{sec:proofoptsc}
\subsection{Linearized Model of the Current Control Loop with the General Linear Filter}
By assuming the continuity condition, we linearize the model of current control loop as
\begin{align}
     \tilde{i}_c [n] = \, & h_0(T_{\text{on}})h_0(T_{\text{off}})\, \tilde{i}_c[n-1] +
     h_0^{'}(t)\bigg\rvert_{t = T_{\text{on}}} h_0(T_{\text{off}}) I_c\, \tilde{t}_{\text{on}}[n] \nonumber \\
     & + \left( u(t) * h(t)\right)\bigg\rvert_{t  = T_{\text{on}}}  \tilde{i}_p[n-1] \nonumber \\
     & + (I_p-m_2T_{\text{off}})\left( u(t) * h(t)\right)^{'} \bigg\rvert_{t = T_{\text{on}}} \tilde{t}_{\text{on}} [n]
     \nonumber \\
     & + \left(m_1t\, u(t) * h(t)\right)^{'}\bigg\rvert_{t = T_{\text{on}}} \tilde{t}_{\text{on}} [n] \nonumber \\ 
     & + \left(w(t)\, u(t) * h(t)\right)^{'}\bigg\rvert_{t = T_{\text{on}}} \tilde{t}_{\text{on}} [n].
\end{align}
For all $t>0^{+}$, $u(t)*h(t)$ and $m_1t\,u(t)*h(t)$ that are differentiable
\begin{align}
     \label{eqn:conv_eq1_app} \left( m_1t\, u(t)*h(t) \right)^{'} &=  m_1u(t)*h(t), \\
     \label{eqn:conv_eq2_app} \left( u(t)*h(t) \right)^{'} &=  h(t).
\end{align}
By substituting (\ref{eqn:conv_eq1_app}), (\ref{eqn:conv_eq2_app}), we have
\begin{align} 
     \tilde{i}_c [n] = & h_0(T_{\text{on}})h_0(T_{\text{off}})\, \tilde{i}_c[n-1] + \left( u(t) * h(t)\right)\bigg\rvert_{t  = T_{\text{on}}}  \tilde{i}_p[n] \nonumber \\
     & + h_0^{'}(t)\bigg\rvert_{t = T_{\text{on}}} h_0(T_{\text{off}}) I_c\,\tilde{t}_{\text{on}}[n] \nonumber \\
     & + (I_p-m_2T_{\text{off}}) h(T_{\text{on}}) \tilde{t}_{\text{on}} [n] \nonumber \\
     & +\left(w(t)\, u(t) * h(t)\right)^{'}\bigg\rvert_{t= T_{\text{on}}} \tilde{t}_{\text{on}} [n].
\end{align}
We denote
\begin{align}
    c_1 = & \left( u(t) * h(t)\right)\bigg\rvert_{t  = T_{\text{on}}} \nonumber \\
    c_2 = & h_0^{'}(t)\bigg\rvert_{t = T_{\text{on}}} h_0(T_{\text{off}}) I_c\,+(I_p-m_2T_{\text{off}}) h(T_{\text{on}}) \nonumber \\
    & +\left(w(t)\, u(t) * h(t)\right)^{'}\bigg\rvert_{t= T_{\text{on}}}. \nonumber \\
\end{align}
The resulting model is
\begin{align}
    \tilde i_p[n] = \; & \tilde i_p[n-1]  +m_1 \, \tilde t_{\text{on}}[n], \\
    \tilde i_c [n] = \; & h_0(T_{\text{on}}) h_0(T_{\text{off}}) \,\tilde i_c [n-1] +  c_1 \, \tilde i_p [n] +  c_2 \, \tilde t_{\text{on}}[n].
\end{align}
\subsection{Linearized Model of the Current Control Loop with First-Order Low Pass Filter}
The zero state response of the first-order low-pass filter given input signal $u(t)$ is
\begin{align} \label{eqn:u_conv_h}
     u(t)*h(t) = \left( 1 - e^{-\frac{t}{\tau}}\right) u(t).
\end{align}
The zero state response of the first-order low-pass filter given input signal $m_1t\,u(t)$ is
\begin{align} \label{eqn:tu_conv_h}
     m_1t\, u(t)*h(t) = m_1\left(t + (e^{-\frac{t}{\tau}}-1)\tau \right) u(t).
\end{align}
By substituting (\ref{eqn:u_conv_h}) and (\ref{eqn:tu_conv_h}) into the period-one equilibrium of (\ref{eqn:filter_nonlin_sys_fb}), the equilibrium can be expressed as
\begin{align} \label{eqn:equic_app}
    I_c = & \frac{\left(\left( I_p -m_2T_{\text{off}} + m_1 \, t + w(t) \right) u(t)*h(t)\right)\bigg\vert_{t=T_{\text{on}}}
    }{1-h_0(T_{\text{on}})h_0(T_{\text{off}})} \nonumber \\
    = & \, \frac{1-e^{-\frac{T_{\text{on}}}{\tau}}}{1-e^{-\frac{T}{\tau}}}(I_p - m_2T_{\text{off}}) + \left(1 + \frac{(e^{-\frac{T_{\text{on}}}{\tau}}-1)}{\frac{T_{\text{on}}}{\tau}} \right) \frac{m_1T_{\text{on}}}{1-e^{-\frac{T}{\tau}}} \nonumber \\
    & + \frac{g(T_{\text{on}}) - \frac{g(0)}{h_0(0)} h_0(T_{\text{on}})}{1-e^{-\frac{T}{\tau}}}.
\end{align}

The resulting model is
\begin{align} \label{eqn:filter_lin_sysx1}
\tilde i_p[n] = \; & \tilde i_p[n-1]  +m_1 \, \tilde t_{\text{on}}[n], \nonumber \\
\tilde i_c [n] = \; & b \,\tilde i_c [n-1] +  c_1 \, \tilde i_p [n] +  c_2 \, \tilde t_{\text{on}}[n],
\end{align}
where
\begin{align}
    c_1 & = 1 - d, \nonumber \\
    c_2 & =  -\frac{b}{\tau}\,I_c + \frac{d}{\tau}\,I_v
    + \frac{\text{d}g(t)}{\text{d}t}\bigg\rvert_{t = T_{\text{on}}} + \frac{e^{-\frac{T_{\text{on}}}{\tau}}}{\tau}g(0), \nonumber \\
    b &=  e^{-\frac{T}{\tau}},  \nonumber \\
    d &= e^{-T_{\text{on}}/\tau}.
\end{align}

\end{appendices}
\newpage

\bibliographystyle{ieeetr}
\bibliography{main}

\begin{thebibliography}{1}

\bibitem{cmpartone2022}
X.~Cui and A.-T. Avestruz, ``Overcoming high frequency limitations of current-mode control using a control conditioning approach -- {Part I}: Modeling and analysis,'' {\em arXiv preprint arXiv:2206.10518}, 2022.

\bibitem{cmparttwo2022}
X.~Cui and A.-T. Avestruz, ``Overcoming high frequency limitations of current-mode control using a control conditioning approach -- {Part II}: Implementation and hardware,'' {\em arXiv preprint arXiv:2206.10523}, 2022.

\bibitem{cui2023lcss}
X.~Cui and A.-T. Avestruz, ``Large-signal stability guarantees for cycle-by-cycle controlled {DC--DC} converters,'' {\em IEEE Control Systems Letters}, vol.~7, pp.~763--768, 2023.

\bibitem{Cui2018a}
X.~Cui and A.-T. Avestruz, ``{A New Framework for Cycle-by-Cycle Digital Control of Megahertz-Range Variable Frequency Buck Converters},'' in {\em 2018 IEEE 19th Workshop on Control and Modeling for Power Electronics (COMPEL)}, (Padova), pp.~1--8, 2018.

\bibitem{xiaofanacc2019}
X.~Cui and A.-T. Avestruz, ``{Switching-Synchronized Sampled-State Space Modeling and Digital Controller for a Constant Off-Time, Current-Mode Boost Converter},'' in {\em 2019 American Control Conference (ACC)}, (Philadelphia), pp.~1--8, 2019.

\bibitem{cui2023tpel}
X.~Cui and A.-T. Avestruz, ``Fast-response variable frequency {DC--DC} converters using switching cycle event-driven digital control,'' {\em IEEE Transactions on Power Electronics}, vol.~38, no.~7, pp.~8190--8207, 2023.

\bibitem{Okuyama2014}
Y.~Okuyama, {\em {Discrete Control Systems}}.
\newblock Springer, Nov 2014.

\end{thebibliography}

\end{document}